\documentclass{article}

\usepackage[utf8]{inputenc}
\usepackage{ae,aecompl,array, amsmath, amssymb, amsthm, color, complexity, verbatim, bbm,amsfonts, hyperref, mathrsfs, graphics,tikz}
\usepackage{stmaryrd}
\usepackage[margin=1.1in]{geometry}
\usepackage{pifont}
\usepackage{bbm}
\usepackage{booktabs}
\usepackage{hhline}

\usepackage{titlesec}

\titleformat{\paragraph}
[runin] 
{\normalfont\normalsize\bfseries} 
{} 
{0pt} 
{\underline} 

\titlespacing*{\paragraph}{0pt}{6pt}{1em} 

\newcommand{\CC}{\ensuremath{\mathscr{C}}}

\newcommand{\SD}{\ensuremath{\mathrm{SD}}}
\newcommand{\one}{\mathbbm{1}}

\newcommand{\Fqn}{\F_q^n}
\newcommand{\carf}[1]{\chi_{#1}}
\newcommand{\car}[2]{\chi_{#1}(#2)}

\usepackage{tikz}

\newcommand{\F}{\mathbb{F}}
\newcommand{\Z}{\mathbb{Z}}

\newcommand{\Iint}[2]{\llbracket #1 , #2 \rrbracket}

\DeclareMathOperator{\tr}{tr}

\newcommand{\trsp}[1]{{#1}^{\intercal}}

\newcommand{\eqdef}{\mathop{=}\limits^{\triangle}}
\newcommand{\Unif}{\leftarrow}

\newcommand{\ie}{\textit{i.e.}}

\newcommand{\COMMENT}[1]{}

\newcommand{\ket}[1]{|#1\rangle}
\newcommand{\bra}[1]{\langle#1|}

\newcommand{\braket}[2]{\langle #1 | #2 \rangle}

\newcommand{\norm}[1]{\left\lVert #1 \right\rVert}

\newcommand{\SIS}{\mbox{SIS}}

\newcommand{\QDP}{\ensuremath{\textrm{QDP}}}

\renewcommand{\aa}{\mathcal{A}}

\newcommand{\mat}[1]{\ensuremath{\boldsymbol{#1}}}

\newcommand{\zerov}{\mathbf{0}}

\newcommand{\cv}{{\mat{c}}}

\newcommand{\ev}{\mat{e}}

\newcommand{\sv}{{\mat{s}}}
\newcommand{\uv}{\mat{u}}

\newcommand{\xv}{{\mat{x}}}
\newcommand{\yv}{{\mat{y}}}
\newcommand{\zv}{{\mat{z}}}

\makeatletter
\newcommand{\ostar}{\mathbin{\mathpalette\make@circled\star}}
\newcommand{\make@circled}[2]{%
	\ooalign{$\m@th#1\smallbigcirc{#1}$\cr\hidewidth$\m@th#1#2$\hidewidth\cr}%
}
\newcommand{\smallbigcirc}[1]{%
	\vcenter{\hbox{\scalebox{0.77778}{$\m@th#1\bigcirc$}}}%
}
\makeatother

\newcommand{\Gm}{\ensuremath{\mathbf{G}}}
\newcommand{\Hm}{\ensuremath{\mathbf{H}}}

\renewcommand{\vec}[1]{\mathbf{#1}}

\newcommand{\wpsi}{\widetilde{\psi}}

\newcommand{\tv}{\vec{t}}
\newcommand{\hf}{\widehat{f}}

\newcommand{\Time}{\mathrm{Time}}

\newcommand{\hu}{\widehat{u}}
\newcommand{\ttau}{\widetilde{\tau}}

\newcommand{\RS}{\ensuremath{\textrm{RS}}}

\newtheorem{theorem}{Theorem}

\newtheorem{definition}{Definition}
\newtheorem{lemma}{Lemma}
\newtheorem{proposition}{Proposition}
\newtheorem{corollary}{Corollary}
\newtheorem{claim}{Claim}

\newtheorem{problem}{Problem}

\usepackage{relsize}
\usepackage{tcolorbox}

\renewcommand{\C}{\mathcal{C}}

\newcommand{\hzsy}{{z_{\sv,\yv}}}
\newcommand{\ohzsy}{\overline{{z_{\sv,\yv}}}}
\newcommand{\Pdec}{\ensuremath{P_{\textrm{Dec}}}}

\renewcommand{\and}{\mbox{ and }}

\renewcommand{\CC}{\ensuremath{\textrm{CC}}}
\newcommand{\ICC}{\ensuremath{\textrm{ICC}}}
\newcommand{\OPI}{\ensuremath{\textrm{OPI}}}

\title{OPI × Soft Decoders}
\author{André Chailloux}
\date{\today}

\begin{document}
	\maketitle
\begin{abstract}
	In recent years, a particularly interesting line of research has focused on designing quantum algorithms for code and lattice problems inspired by Regev’s reduction. The core idea is to use a decoder for a given code to find short codewords in its dual. For example, Jordan et al.~\cite{JSW+24} demonstrated how structured codes can be used in this framework to exhibit some quantum advantage. In particular, they showed how the classical decodability of Reed–Solomon codes can be leveraged to solve the Optimal Polynomial Intersection (OPI) problem quantumly. This approach was further improved by Chailloux and Tillich~\cite{CT25} using stronger soft decoders, though their analysis was restricted to a specific setting of OPI.
	
	In this work, we reconcile these two approaches. We build on a recent formulation of the reduction by Chailloux and Hermouet~\cite{CH25} in the lattice-based setting, which we rewrite in the language of codes. With this reduction, we show that the results of Jordan et al. can be recovered under Bernoulli noise models, simplifying the analysis. This characterization then allows us to integrate the stronger soft decoders of Chailloux and Tillich into the OPI framework, yielding improved algorithms.
\end{abstract}
\tableofcontents
\newpage
\section{Introduction} Regev's reduction~\cite{Reg05} is one of the foundational results of lattice-based cryptography and is a quantum reduction between the Short Integer Solution (SIS) problem and the Learning With Errors (LWE) problem. The reduction is actually to a somewhat easier problem than LWE, where the errors are in quantum superposition. This has been leveraged in~\cite{BKSW18} where they showed relations between these problems and variants of the Dihedral Coset Problem. Then, Chen, Liu, and Zhandry~\cite{CLZ22} showed how to use these ideas not to obtain complexity reductions but rather to construct new quantum algorithms. They give an efficient quantum algorithm for $\SIS_\infty$ in a somewhat contrived parameter regime. While this algorithm has been dequantized~\cite{KOW25}, it paved the way for several new quantum algorithms. Notably, Yamakawa and Zhandry~\cite{YZ24} showed cases that exhibit a provable quantum advantage with this approach in the Random Oracle Model. There has been a series of recent works that studied the algorithmic aspects as well as extensions to coding-theoretic settings~\cite{DRT24,DFS24,CT24,JSW+24,CHL+25,CT25,BCT25,CH25,Hill25,GJ25,SLS+25}.

 In particular, Jordan et al.~\cite{JSW+24} studied structured codes in order to propose simple problems with a large quantum advantage. They looked at particular instances of LDPC codes and Reed--Solomon codes, and the resulting problems they solved can be phrased in terms of the MAX LINSAT problem (from LDPC codes) or the Optimal Polynomial Interpolation (OPI) problem (from Reed--Solomon codes). Chailloux and Tillich~\cite{CT25} presented a more general reduction for these algorithms that tolerates errors in the decoders. This allowed them to use more powerful decoders for Reed--Solomon codes. However, they could not fully integrate these decoders into the framework of~\cite{JSW+24}. In~\cite{CH25}, the authors extended the main reduction of~\cite{CT25} to make it both more general and easier to use in the context of Euclidean lattices.
 
 In this work, we rephrase the reduction of~\cite{CH25} in terms of codes. In light of this new reduction, we reprove and extend some of the results of~\cite{JSW+24} and~\cite{CT25} in an arguably simpler fashion and show how to use the more powerful decoders studied in~\cite{CT25} within the framework of~\cite{JSW+24}, thus combining the strengths of both approaches. We recover and extend the results of the Decoded Quantum Interferometry algorithm with explicit error functions $f$, and we also extend the results of~\cite{CT25} to more general instances of $\ICC$. More precisely: 
 \begin{enumerate} \setlength\itemsep{-0.08cm}
 	\item For MAX Linsat, we recover the same results as in~\cite{JSW+24} with product error functions $f$. 
 	\item In the case of OPI, we recover the results of Decoded Quantum Interferometry with product error functions $f$. With this new characterization, we can incorporate the stronger decoders of~\cite{CT25} in this framework, which gives improved results for the OPI problem. 
 \end{enumerate} 
 All of these results are obtained in an arguably simpler framework which can easily encompass other choices of codes and decoding problems.

\section{Preliminaries}
\subsection{Basic notations and probability theory} A probability function $p$ on a set $U$ is a function $p : U \rightarrow \mathbb{R}_+$ such that $\sum_{x \in U} p(x) = 1$. For a function $f : U \rightarrow \mathbb{C}$, we have $\norm{f}_2 = \sqrt{\sum_{x \in U} |f(x)|^2}$. We will work in a finite field $\F_q$ for any $q$ that is a prime power. Lowercase bold letters $\xv,\yv$ denote (row) vectors with coefficients in $\F_q$, and uppercase bold letters $\Gm,\Hm$ denote matrices with coefficients in $\F_q$. 

 \begin{claim} Let \( X_1, \ldots, X_n \) be independent Bernoulli random variables with \( \Pr[X_i = 0] = 1 - \tau \) and \( \Pr[X_i = 1] = \tau \) for all \( i \). Let $S_n = \sum_{i = 1}^n X_i$. For any $\ttau < \tau$, \[ \Pr[ S_n \le \ttau n] \leq 2e^{-2n(\tau - \ttau)^2}. \] \end{claim} 
\subsection{Finite fields and Fourier transform.} For $\xv = (x_i)_{i \in \Iint{1}{n}} \in \F_q^n$ and $\yv = (y_i)_{i \in \Iint{1}{n}} \in \F_q^n$, we define their inner product $\xv \cdot \yv = \sum_i x_i y_i$, where the sum and multiplication are over $\F_q$. We now give a brief overview of the characters of $\F_q$, following~\cite{CT25}. \begin{definition} Let $q = p^s$ for a prime integer $p$ and an integer $s \ge 1$. The characters of $\F_q$ are the functions $\carf{y} : \F_q \rightarrow \mathbb{C}$ indexed by elements $y \in \F_q$, defined as follows: \begin{eqnarray*} \car{y}{x} & \eqdef & e^{\frac{2i \pi \tr(x \cdot y)}{p}}, \quad \text{with} \\ \tr(a) & \eqdef & a + a^p + a^{p^2} + \dots + a^{p^{s-1}}. \end{eqnarray*} where the product $x \cdot y$ corresponds to multiplication in $\F_q$. We extend the definition to vectors $\xv,\yv \in \F_q^n$ as follows: $$ \car{\yv}{\xv} \eqdef \prod_{i = 1}^n \car{y_i}{x_i} = e^{\frac{2i \pi \tr(\xv \cdot \yv)}{p}}.$$ \end{definition} 

Notice that inner product $\xv \cdot \yv$ is equal to the vector multiplication $\xv \trsp{\yv}$. This means that for $\xv \in \F_q^k, \yv \in \F_q^n$ and $\Gm \in \F_q^{k \times n}$, we have  $\xv \Gm \cdot \yv = \xv \cdot \yv \trsp{\Gm}$ hence $\chi_{\yv}(\xv \Gm) = \chi_{\yv \trsp{\Gm}}(\xv)$. 
 Characters have many desirable properties that we can use for our calculations. \begin{proposition}\label{prop:characters} The characters $\carf{\yv} : \Fqn \rightarrow \mathbb{C}$ have the following properties: \begin{enumerate}\setlength\itemsep{-0.2em} \item (Group homomorphism). For all $\yv \in \Fqn$, $\carf{\yv}$ is a group homomorphism from $(\Fqn,+)$ to $(\mathbb{C},\cdot)$, meaning that for all $\xv, \xv' \in \Fqn$, $\car{\yv}{\xv + \xv'} = \car{\yv}{\xv} \cdot \car{\yv}{\xv'}$. \item (Symmetry). For all $\xv, \yv \in \Fqn$, $\car{\yv}{\xv} = \car{\xv}{\yv}$. \item (Orthogonality of characters). The characters are orthogonal functions: for all $\xv, \xv' \in \Fqn$, \\ $\sum_{\yv \in \Fqn} \car{\yv}{\xv}\overline{\car{\yv}{\xv'}} = q^n \delta_{\xv,\xv'}$. In particular, $\sum_{\yv \in \Fqn} |\car{\yv}{\zerov}|^2 = q^n$ and for all $\xv \in \Fqn\setminus\{\zerov\}$, $\sum_{\yv \in \Fqn} \car{\yv}{\xv} = 0$. \end{enumerate} \end{proposition} \begin{definition} For a function $f : \F_q^n \rightarrow \mathbb{C}$, we define its Fourier transform as $\hf(\xv) = \frac{1}{\sqrt{q^n}} \sum_{\yv \in \F_q^n} \chi_{\xv}(\yv) f(\yv)$. \end{definition} \begin{claim}[Parseval's Identity] For any $f : \F_q^n \rightarrow \mathbb{C}$, we have $\norm{f}_2 = \norm{\hf}_2$. \end{claim} \begin{claim} Let $\C$ be a $q$-ary linear code. Then $$ \sum_{\cv \in \C} \chi_{\yv}(\cv) = \begin{cases} |\C| & \text{if } \yv \in \C^{\bot},\\ 0 & \text{otherwise.} \end{cases}$$ \end{claim}
 \subsection{Linear Codes} A $q$-ary linear code $\C$ of dimension $k$ and length $n$ is characterized by a full-rank generating matrix $\Gm \in \F_q^{k \times n}$ or equivalently by a full-rank parity-check matrix $\Hm \in \F_q^{(n-k) \times n}$, and we write \[ \C = \{\xv \Gm : \xv \in \F_q^k\} \quad \text{or} \quad \C = \{\yv \in \F_q^n : \Hm \trsp{\yv} = \zerov\}.\] Each code $\C$ has an associated dual code $\C^\bot = \{\yv \in \F_q^n : \forall \cv \in \C, \ \yv \cdot \cv = 0\}$. Any generating matrix of $\C$ is a parity-check matrix of $\C^{\bot}$. For a string $\uv \in \F_q^{n-k}$ called syndrome, we define the coset $\C_{\uv} = \{\yv \in \F_q^n : \Hm \trsp{\yv} = \trsp{\uv}\}$. For $\uv \in \F_q^{k}$, called a dual syndrome, we also have the dual coset $\C^\bot_{\uv} = \{\yv \in \F_q^n : \Gm \trsp{\yv} = \trsp{\uv}\}$. We will be particularly interested in full-support Reed--Solomon codes. The full-support Reed--Solomon code of length $n$ and dimension $k$ is the code $$ \RS_k = \{\left(P(\alpha_1),\dots,P(\alpha_q)\right) : P \in \F_q[X],\ \deg(P) < k\}, $$ where $\{\alpha_1,\dots,\alpha_q\} = \F_q$. Note that for these codes we have $n = q$. An important property of the full-support Reed--Solomon code is that its dual is also a full-support Reed--Solomon code. \begin{claim} $\RS_k^\bot = \RS_{n-k}$. \end{claim} \subsection{Code problems and OPI} We will be interested in the following code problems. \begin{problem}[Syndrome Decoding Problem $\SD(\C,p)$] ~\\ \textbf{Given:} $(\C,\cv+\ev,p)$ where $\C$ is a $q$-ary linear code of dimension $k$ and length $n$ for which we have a description (e.g., via a generating matrix $\Gm$), $\cv \Unif \C$, and $\ev$ is distributed according to $p : \F_q^n \rightarrow \mathbb{R}_+$. \\ \textbf{Goal:} Find $\cv$. \end{problem} This problem has been extended to the case of codes in~\cite{CT24} where the errors are considered in quantum superposition, \begin{problem}[Quantum Decoding Problem $\QDP(\C,f)$] ~\\ \textbf{Given:} $(\C,\ket{\psi_\cv},f)$ where $\C$ is a $q$-ary linear code of dimension $k$ and length $n$ for which we have a description, $\cv \Unif \C$, and $\ket{\psi_\cv} = \sum_{\ev \in \F_q^n} f(\ev) \ket{\cv + \ev}$ for a function $f : \F_q^n \rightarrow \mathbb{C}$ with $\norm{f}_2 = 1$. \\ \textbf{Goal:} Find $\cv$. \end{problem} Notice that $\QDP(\C,f)$ is easier than the $\SD(\C,|f|^2)$ problem, as one can just measure $\ket{\psi_\cv}$ in the computational basis in order to recover an instance of $\SD(\C,|f|^2)$. \begin{problem}[Constrained Codeword Problem $\CC(\C,T)$] ~\\ \textbf{Given:} The description of a $q$-ary linear code $\C$ of length $n$ and dimension $k$, and a set $T \subseteq \F_q^n$. \\ \textbf{Goal:} Find $\yv \in \C \cap T$. \end{problem} We will actually be more interested in the inhomogeneous variant of the above problem, where we want to find $\yv \in \C_{\uv} \cap T$ for a randomly chosen dual syndrome $\uv$. \begin{problem}[Inhomogeneous Constrained Codeword Problem $\ICC(\C,T)$] ~\\ \textbf{Given:} The description of a $q$-ary linear code $\C$ of length $n$ and dimension $k$ as well as a parity-check matrix $\Hm \in \F_q^{(n-k) \times n}$ of $\C$, and a random syndrome $\uv \in \F_q^{n-k}$. \\ \textbf{Goal:} Find $\yv \in \C_{\uv} \cap T$. \end{problem} For this problem, we will quantify the success probability on average over the syndrome $\uv$. \begin{definition} Let $\aa(\C,T,\uv)$ be an algorithm that solves $\ICC(\C,T)$ for a fixed $\uv$, and let $p_{\uv}$ be the probability that it finds $\yv \in \C_{\uv} \cap T$. We say that $\aa$ solves $\ICC(\C,T)$ with probability $p$ if $\E_{\uv \Unif \F_q^{n-k}}[p_{\uv}] = p$. \end{definition} Finally, we present the Optimal Polynomial Interpolation problem. \begin{problem}[Optimal Polynomial Interpolation Problem $\OPI(k,q,\{S_i\}_{i \in \F_q},\tau)$] ~\\ \textbf{Given:} Positive integers $k,q > k$ and subsets $S_i \subseteq \F_q$ for $i \in \F_q$, and a real number $\tau \in [0,1]$. A random string $\xv = (x_i)_{i \in \F_q} \in \F_q^q$. \\ \textbf{Goal:} Find $P \in \F_q[X]$ with $\deg(P) < k$ such that $|\{i \in \F_q : P(i) + x_i \in S_i\}| \ge \tau q$. \end{problem} One can remark that the OPI problem is a special instantiation of the $\ICC$ problem. \begin{claim}\label{Claim:5} The $\OPI(k,q,\{S_i\}_{i \in \F_q},\tau)$ problem and the $\ICC(\RS_k,T)$ problem are equivalent, with $$ T = \left\{(\yv_i)_{i \in \F_q^q} : \left|\{i : y_i \in S_i \}\right| \ge \tau q\right\}.$$ \end{claim} \begin{proof} Fix $k,q$ and let $\Hm_k \in \F_q^{q-k}$ be the parity matrix of $\RS_k$. Assume we start from an $\OPI(k,q,\{S_i\},\tau)$ instance with random string $\xv$. Let $\uv = \Hm_k \xv$ and run an algorithm for $\ICC(\RS_k,T)$ with 
 	$$ T = \left\{(\yv_i)_{i \in \F_q^q} : \left|\{i : y_i \in S_i \}\right| \ge \tau q\right\}.$$ We obtain $\yv \in (\RS_k)_{\uv} \cap T$. In particular, $\Hm_k(\yv - \xv) = \zerov$ hence $\yv - \xv \in \RS_k$. Let $\yv' = \yv - \xv \in \RS_k$. We write $\yv' = (y'_i)_{i \in \F_q} = (P(i))_{i \in \F_q}$ for some polynomial $P \in \F_q[X]$ such that $\deg(P) < k$ which can be easily recovered from $\yv'$ using polynomial interpolation. This $P$ will be a solution to our $\OPI$ instance. Indeed, using $P(i) + x_i = y_i$, we have $$ |\{i \in \F_q : P(i) + x_i \in S_i\}| = |\{i \in \F_q : y_i \in S_i\}| \ge \tau q, \quad \text{since } \yv \in T.$$ We now prove the other direction. Start from an $\ICC(\RS_k,T)$ instance with a choice of $\{S_i\}$ and $\tau$ which defines $T$ as above; and a random $\uv \in \F_q^{n-k}$. Let $\xv \in \F_q^n$ be a random element of $\C_{\uv}$. This can be done using Gaussian Elimination after randomly permuting the lines. Run the $\OPI(k,q,\{S_i\},\tau)$ solver to obtain $P \in \F_q[X]$ with $\deg(P) < k$ such that $|\{i \in \F_q : P(i) + x_i \in S_i\}| \ge \tau q$. Let $\yv = (y_i)_{i \in \F_q}$ with $y_i = P(i)$. Notice that $\yv \in \RS_k$. This implies $\yv + \xv \in \C_{\uv} \cap T$ so we can output $\xv + \yv$ to obtain a solution of $\ICC(\RS_k,T)$. \end{proof} \subsection{The~\cite{CH25} Reduction Theorem and applications} In~\cite{CH25}, the authors present a reduction from the ISIS problem to the S-$\ket{\text{LWE}}$ problem which can be lattice-based variants of the $\ICC$ problem and the $\QDP$ problem respectively. This theorem can be translated to the setting of codes as follows: \begin{theorem}\label{Theorem:Main} Let $\C$ be a $q$-ary linear code of dimension $k$ and length $n$ and let $\Gm \in \F_q^{k \times n}$ be a generating matrix of $\C$. Let $T \subseteq \F_q^n$. Let $f : \F_q^n \rightarrow \mathbb{C}$ with $\norm{f}_2 = 1$. Assume that \begin{enumerate}\setlength\itemsep{-0.3em} \item We have a quantum algorithm $\aa_{\QDP}$ that solves $\QDP(\C,f)$ in time $\Time_{\QDP}$ and succeeds with probability $P_{Dec}$. \item The state $\sum_{\ev \in \F_q^n} f(\ev) \ket{\ev}$ is constructible in time $T_{Sampl}$. \item $\sum_{\yv \in T}|\hf(\yv)|^2 = 1 - \eta$. \end{enumerate} Then there exists a quantum algorithm that solves $\ICC(\C^\bot,T)$ with probability $$ P \ge P_{Dec}(1-\eta) - 2\sqrt{\eta P_{Dec}(1-P_{Dec})},$$ and runs in time $$\Time = O\left(\frac{1}{P_{Dec}} \left(\Time_{\QDP} + \Time_{Sampl}\right) + \poly(n,\log(q))\right).$$ \end{theorem} \begin{proof} The proof of~\cite{CH25} can almost verbatim be translated by replacing the lattice-based concepts to the code-based concepts. The only difference is that we work on $\F_q$ while~\cite{CH25} worked on $\Z_q$ but this is a very minor difference and doesn't alter the proof. For completeness, we rewrite the proof of~\cite{CH25} in the code-based setting in Appendix~\ref{Appendix:MainProof}. \end{proof} \COMMENT{ This theorem can be schematically represented as follows: \\ \begin{tikzpicture} \tikzstyle{boxA}=[draw, rectangle, align=center, text width=5.3cm, minimum height=1cm] \tikzstyle{boxB}=[draw, rectangle, align=center, text width=6.7cm, minimum height=1cm] \node[boxA] (box1) at (0,0) {There is a quantum algorithm that solves $\QDP(\C,f)$ wp. $P_{\QDP}$ and $Pr_{\yv \Unif |\hf|^2}[\yv \in T] \ge 1 - \eta$.}; \draw[->, thick] (box1.east) -- node[above] {Theorem~\ref{Theorem:Main}} ++(2.6,0) coordinate (midpoint); \node[boxB, anchor=west] (box2) at (midpoint) {There is a quantum algorithm that solves $\ICC(\C^\bot,T)$ wp. \\ $p \ge P_{Dec}(1-\eta) - 2\sqrt{\eta P_{Dec}(1-P_{Dec})}$.}; \end{tikzpicture}$ \ $ \\ } This theorem shows how to construct a quantum algorithm for $\ICC(\C^\bot,T)$ given an algorithm for $\QDP(\C,f)$ Here, the third item of our theorem allows us to study very simple product functions $f$, and we can handle distribution tails arising from these product error functions. 

\section{Choice of error function} Our goal is to construct functions for which we can apply Theorem~\ref{Theorem:Main}. For the set $T$, we pick sets $\{S_i\}_{i \in \F_q}$ where each $S_i \subseteq \F_q$ as well as a threshold $\ttau \in (0,1)$. We consider the set $$T = \{\yv = (y_i)_{i \in \Iint{1}{n}} : |\{i : y_i \in S_i\}| \ge \ttau n\}.$$ The case $\ttau = 1$ corresponds to~\cite{CT25}, whereas arbitrary $\ttau$ was studied in~\cite{JSW+24}. 

\begin{proposition}\label{Proposition:2} Let subsets $S_i \subseteq \F_q$ of a fixed same size for each $i \in \F_q$. Let $\rho = \frac{|S_i|}{q}$ which is independent of $i$. Let a threshold $\ttau \in (0,1)$ and $T_{\ttau} = \{\yv = (y_i)_{i \in \Iint{1}{n}} : |\{i : y_i \in S_i\}| \ge \ttau n\}$. Let $\tau = \ttau + n^{-1/3} = \ttau + o(1)$ and consider the functions $u_i : \F_q \rightarrow \mathbb{C}$ such that \[ \hu_i(\alpha) = \left\{ \begin{array}{cl} \sqrt{\frac{\tau}{|S_i|}} & \textrm{ if } \alpha \in S_i \\ \sqrt{\frac{1 - \tau}{q - |S_i|}} & \textrm{ if } \alpha \notin S_i \end{array}\right.\]
	Finally, let $f = \otimes_{i = 1}^n u_i$. We have
	 \begin{enumerate} 
		\item $\sum_{\yv \in T_{\ttau}} |\hf(\yv)|^2 = 1 - \eta$ with $\eta = negl(n)$.
		\item $\forall i \in \F_q, \ |u_i(0)|^2 = \left(\sqrt{\tau \rho} + \sqrt{(1-\tau)(1 - \rho)}\right)^2$, where $\tau = \ttau + o(1)$.
\end{enumerate} 
\end{proposition}

 \begin{proof} 
 	 The relation $f = \otimes_{i = 1}^n u_i$ implies $\hf = \otimes_{i = 1}^n \hu_i$, which means we can write \[ \sum_{\yv \in \F_q^n} \hf(\yv) \ket{\yv} = \bigotimes_{i \in \F_q} \left(\sum_{\alpha \in \F_q} \hu_i(\alpha) \ket{\alpha}\right) = \bigotimes_{i \in \F_q} \left(\sum_{\alpha \in S_i} \sqrt{\frac{\tau}{|S_i|}} \ket{\alpha} + \sum_{\alpha \notin S_i} \sqrt{\frac{1 - \tau}{q - |S_i|}} \ket{\alpha}\right) .\] From Hoeffding's inequality, we have \begin{align} \sum_{\yv \in T} |\hf(\yv)|^2 \ge 1 - e^{-2n(\tau - \ttau)^2} = 1 - negl(n). \end{align} For the second point, we use $u_i(\alpha) = \frac{1}{\sqrt{q}} \sum_{\beta \in \F_q} \chi_{- \alpha}(\beta) \hu_i(\beta).$ In particular, \begin{align*} |u_i(0)|^2 & = \frac{1}{q} \left|\sum_{\beta \in \F_q} \hu_i(\beta)\right|^2 = \frac{1}{q} \left(\sum_{\beta \in S_i} \sqrt{\frac{\tau}{|S_i|}} + \sum_{\beta \notin S_i} \sqrt{\frac{1 - \tau}{q - |S|}}\right)^2 = \frac{1}{q}\left(\sqrt{\tau |S_i|} + \sqrt{(1-\tau)(q-|S_i|)}\right)^2 \\ & = \left(\sqrt{\tau \rho} + \sqrt{(1-\tau)(1 - \rho)}\right)^2 \end{align*} since $\rho = \frac{|S_i|}{q}$ for each $i \in \F_q$. \end{proof} We suppose each $|S_i|$ has the same size as it is the standard setting and it is easier to assess the performance of decoders with this kind of error functions, but the above theorem can be easily generalized to the case where the $|S_i|$ do not have the same size. 

\section{Applications}
\subsection{The case of binary codes}
We first apply our theorem to the binary setting
\begin{proposition}
	Let $\C$ be a binary linear code of length $n$. Assume there exists an efficient algorithm $A_{Dec}$ that can decode a $t$ fraction of errors with probability $P_{Dec}$. Then there exists a quantum algorithm that solves efficiently $\ICC(\C^{\bot},T_{\ttau})$ where 
	$$ T_{\ttau} = \{\yv \in \F_2^n : |\yv| \le (1-\ttau) n\} = \{\yv = (y_i)_{i \in \Iint{1}{n}} : |\{i : y_i \in S\}| \ge \ttau n\}, \text{ for } S = \{0\}$$
	where $\ttau = \tau - n^{1/3} = \tau - o(1)$ and 
	$\tau = \frac{1}{2} + \sqrt{t(1-t)} = \left(\sqrt{\frac{t}{2}} + \sqrt{\frac{1-t}{2}}\right)^2$
	This quantum algorithm succeeds with probability $P_{Dec} - negl(n)$.
\end{proposition}
\begin{proof}
	In this setting, saying that $\aa_{Dec}$ solves a $t$-fraction of errors means it can solve $\SD(\C,p^{\otimes n })$ with $p(0) = 1-t$ and $p(1) = t$. Fix $u : \F_2^n \rightarrow \mathbb{C}$ with $u(0) = \sqrt{1-t}$ and $u(1) = \sqrt{t}$. Let $f = u^{\otimes n}$. $\aa_{Dec}$ solves $\SD(\C,u^{\otimes})$, hence $\QDP(\C,f)$ with probability $P_{Dec}$.  We can compute 
	\begin{align*}
		\hu(0) & = \frac{1}{\sqrt{2}}\left(u(0) + u(1)\right) =  \left(\sqrt{\frac{t}{2}} + \sqrt{\frac{1-t}{2}}\right) = \sqrt{\tau} \\
		\hu(1) & = \frac{1}{\sqrt{2}}\left(u(0) - u(1)\right) =  \left(\sqrt{\frac{t}{2}} - \sqrt{\frac{1-t}{2}}\right) = \sqrt{1 - \tau} 
	\end{align*}
	We use Proposition~\ref{Proposition:2}, with $q=2$ and each $S_i = \{0\}$ to obtain that 
	$$ \sum_{yv \in T_{\ttau}} |\hf(\yv)|^2 = 1 - \eta \text{ with } \eta = negl(n).$$
	We can therefore conclude using Theorem~\ref{Theorem:Main} that there exists an efficient quantum algorithm for $\ICC(\C^\bot,T_{\ttau})$ that succeeds with probability $P_{Dec} - negl(n)$.
\end{proof}

\begin{corollary}[Recovering results of~\cite{JSW+24}]
Consider a code $\C$ for which we can decode on average a fraction $t = \frac{6350}{50000}$ of errors. Then, we can solve  efficiently $\ICC(\C^{\bot},T)$ where 
$$ T = \{\yv \in \F_2^n : |\yv| \le (1-\ttau) n\}, \text{ for } \ttau \approx 0.833\footnote{In~\cite{JSW+24}, they claim only $\ttau \approx 0.831$ which comes from a non-asymptotic analysis, but actually arises from the same calculations.}.$$
\end{corollary}
The $\ICC(\C^\bot,T)$ solved in this corollary corresponds exactly to the Max XORSAT instance solved in~\cite{JSW+24}, for the choice of LDPC code specified therein.

\subsection{OPI}
Now, we consider a prime $q$ and sets $S_i \subseteq \F_q$ of the same size for $i \in \F_q$. Let $\rho = \frac{|S_i|}{q}$. The idea is to work on full rank Reed--Solomon codes. We know from Claim~\ref{Claim:5} that solving $\ICC(\RS_k,T)$ for a well chosen $T$ is equivalent to solving an OPI problem. Recall also that for the case of full rank Reed--Solomon codes, the length of the code $n$ is equal to the alphabet size $q$ but we will still use the two variables for readability. We first recall the performance of Reed--Solomon decoders
\begin{proposition}\label{Proposition:Decoders} We consider the different decoders applied on $\RS_{n-k}$, {\ie} the full support Reed-Solomon codes of dimension $(n-k)$ and length $n$. 
Let $p$ be a probability function on $\F_q^n$ with $p = \otimes_{i = 1}^n p_i$ such that $\forall i, \ p_i(0) = 1-t$. 
	\begin{itemize}\setlength\itemsep{-0.2em}
		\item The \emph{Berlekamp-Welch} algorithm solves $\SD(\RS_{n-k},p)$ wp. at least $\frac{1}{\poly(n)}$ for $\frac{k}{n} \ge 2t = 2(1-p_i(0))$.
		\item The \emph{Guruswami-Sudan} algorithm solves $\SD(\RS_{n-k},p)$ wp. at least $\frac{1}{\poly(n)}$ for $\frac{k}{n} \ge 1 - (1-t)^2 = 1 - p_i^2(0)$.
	\end{itemize}
\end{proposition}
\begin{proof}
	The Berlekamp-Welch algorithm~\cite{Wel83} can solve the decoding problem on $\RS_{n-k}$ as long as the number of errors is at most $t_0 = \lfloor \frac{k}{2} \rfloor$. If we take the error distribution $p$ where each $p_i(0) \ge 1 - \frac{k}{2n}$ then the number of errors will smaller than $t_0 \approx n(1 - p_i(0))$ with constant probability. We can rewrite the inequality as $\frac{k}{n} \ge 2(1 - p_i(0))$. 
	
	The Guruswami-Sudan algorithm~\cite{GS98} can solve the decoding problem on $\RS_{n-k}$ as long as the number of errors is at most $t_0 = n - \sqrt{n(n-k)}$. If we take the error distribution $p$ where each $p_i(0) \ge \sqrt{\frac{n -k}{n}}$, then the number of errors will smaller than $t_0 = n(1 - p_i(0))$ with constant probability. We can rewrite the inequality as $\frac{k}{n} \ge 1 - p_i^2(0)$.
\end{proof}
In the case $p$ is a product function of the same distribution, so $p = v^{\otimes n}$, we have the following improvement over the Guruswami Sudan algorithm. 
\begin{proposition}\label{Proposition:KV}[\cite{MT17}] Let $v$ be a probability function on $\F_q$.
		The \emph{Koetter-Vardy} algorithm solves $\SD(\RS_{n-k},v^{\otimes n})$ wp. at least $\frac{1}{\poly(n)}$ for $\frac{k}{n} \ge 1 - \sum_{\alpha \in \F_q} v^2(\alpha)$.
\end{proposition}

\subsubsection{Using the Berlekamp-Welch decoder}

\begin{theorem}[Recovering results of~\cite{JSW+24}]
	 Let positive integers $k,q$ and $n=q$. Let $S_1,\dots,S_q \subseteq \F_q$ of the same length with $\rho = \frac{|S_i|}{q}$. Let $T_{\ttau} = \{\yv = (y_i)_{i \in \Iint{1}{n}} : |\{i : y_i \in S_i\}| \ge \ttau q\}$ for any $\ttau \in (0,1)$. There exists a quantum algorithm that solves $\ICC(\RS_{k},T_{\ttau})$ with probability $\frac{1}{\poly(n)}$ for $\ttau = \tau - n^{-1/3} = \tau - o(1)$ where $\tau$ is any real such that 
	\begin{align}\label{Eq:Thm3} 1-\frac{k}{2n} \le \left(\sqrt{\tau \rho} + \sqrt{(1-\tau)(1-\rho)}\right)^2.\end{align}
	In particular, for $\frac{k}{n} = 0.1$ and $\rho = 0.5$, we recover numerically that $\tau \approx 0.7179$ saturates the above inequality.
\end{theorem}
\begin{proof} Fix $\tau$ that satisfies Equation~\ref{Eq:Thm3}. We consider the functions
	$u_i : \F_q \rightarrow \mathbb{C}$ such that \[ \hu_i(\alpha) = \left\{ \begin{array}{cl} \sqrt{\frac{\tau}{|S_i|}} & \textrm{ if } \alpha \in S_i \\ \sqrt{\frac{1 - \tau}{q - |S_i|}} & \textrm{ if } \alpha \notin S_i \end{array}\right.\]
	Also, let $f = \otimes_{i = 1}^n u_i$ and $p = |f|^2 = \otimes_{i = 1}^n p_i$ with $p_i = |u_i|^2$. First, we have from Proposition~\ref{Proposition:2} that 
	\begin{align*}
		p_i(0) = |u_i(0)|^2 = \left(\sqrt{\tau \rho} + \sqrt{(1-\tau)(1-\rho)}\right)^2.
	\end{align*}
	We therefore obtain 
	$1 - \frac{k}{2n} \le p_i(0) \text{ which implies } \frac{k}{n} \ge 2(1-p_i(0)).$ We can use Proposition~\ref{Proposition:Decoders} to say that there exists an efficient algorithm to solve $\SD(\RS_{n-k},p)$, hence $\QDP(\RS_{n-k},f)$ with probability $\frac{1}{\poly(n)}$. This gives the first condition required in Theorem~\ref{Theorem:Main}. Since states $\sum_{\alpha \in \F_q} \hu_i(\alpha) \ket{\alpha}$ are efficiently computable, the state $\sum_{\ev \in \F_q^n} f(\ev) \ket{\ev}$ is also efficiently computable. For the third condition, we know from Proposition~\ref{Proposition:2} that 
	$$ \sum_{\yv \in T_{\ttau}} |\hf(\yv)|^2 = 1 - \eta \text{ with } \eta = negl(n).$$
	We use Theorem~\ref{Theorem:Main} to conclude that there exists an efficient quantum algorithm that solves $\ICC(\RS_{n-k}^\bot,T_{\ttau}) = \ICC(\RS_{k},T_{\ttau})$ efficiently with success probability $\frac{1}{\poly(n)} - negl(n) = \frac{1}{\poly(n)}$.
\end{proof}

\subsubsection{Using the Guruswami-Sudan decoder}
We can reproduce the above argument by replacing the Berlekamp-Welch algorithm with the Guruswami-Sudan algorithm. This gives the following theorem

\begin{theorem}[Generalizing results of ~\cite{CT25}, Guruswami-Sudan decoder]
		Let positive integers $k,q$ and $n=q$. Let $S_1,\dots,S_q \subseteq \F_q$ of the same length with $\rho = \frac{|S_i|}{q}$. Let $T_{\ttau} = \{\yv = (y_i)_{i \in \Iint{1}{n}} : |\{i : y_i \in S_i\}| \ge \ttau q\}$ for any $\ttau \in (0,1)$. There exists a quantum algorithm that solves $\ICC(\RS_{k},T_{\ttau})$ with probability $\frac{1}{\poly(n)}$ for $\ttau = \tau - n^{-1/3} = \tau - o(1)$ where $\tau$ is any real such that  
	\begin{align}\label{Eq:Thm4} 1-\frac{k}{n} \le \left(\sqrt{\tau \rho} + \sqrt{(1-\tau)(1-\rho)}\right)^4. \end{align}
	In particular, for $\frac{k}{n} = \frac{3}{4}$ and $\rho = 0.5$, the above is satisfied even with $\tau = 1$, recovering the results~\cite{CT25}.
\end{theorem}
The proof is the same as the previous one but with a different decoder. We reproduce it here for completeness and clarity.
\begin{proof}
	Fix $\tau$ that satisfies Equation~\ref{Eq:Thm3}. We consider the functions
	$u_i : \F_q \rightarrow \mathbb{C}$ such that \[ \hu_i(\alpha) = \left\{ \begin{array}{cl} \sqrt{\frac{\tau}{|S_i|}} & \textrm{ if } \alpha \in S_i \\ \sqrt{\frac{1 - \tau}{q - |S_i|}} & \textrm{ if } \alpha \notin S_i \end{array}\right.\]
	Also, let $f = \otimes_{i = 1}^n u_i$ and $p = |f|^2 = \otimes_{i = 1}^n p_i$ with $p_i = |u_i|^2$. First, we have from Proposition~\ref{Proposition:2} that 
	\begin{align*}
		p_i(0) = |u_i(0)|^2 = \left(\sqrt{\tau \rho} + \sqrt{(1-\tau)(1-\rho)}\right)^2.
	\end{align*}
	We therefore obtain $1- \frac{k}{n} \le p_i^2(0)$ which implies $\frac{k}{n} \ge 1 - p_i^2(0)$. We can use the second bullet of Proposition~\ref{Proposition:Decoders} to say that there exists an efficient algorithm to solve $\SD(\RS_{n-k},p)$, hence $\QDP(\RS_{n-k},f)$ with probability $\frac{1}{\poly(n)}$. This gives the first condition required in Theorem~\ref{Theorem:Main}. Since states $\sum_{\alpha \in \F_q} \hu_i(\alpha) \ket{\alpha}$ are efficiently computable, the state $\sum_{\ev \in \F_q^n} f(\ev) \ket{\ev}$ is also efficiently computable. For the third condition, we know from Proposition~\ref{Proposition:2} that 
	$$ \sum_{\yv \in T_{\ttau}} |\hf(\yv)|^2 = 1 - \eta \text{ with } \eta = negl(n).$$
	We use Theorem~\ref{Theorem:Main} to conclude that there exists an efficient quantum algorithm that solves $ \ICC(\RS_{k},T_{\ttau})$ efficiently with success probability $\frac{1}{\poly(n)}$.
\end{proof}
	
\subsubsection{Using the Koetter-Vardy decoder in the case of the infinity norm}
 A natural case is when each $S_i = \Iint{-z}{z}$ (recall that we write $\F_q = \Iint{-\lfloor\frac{q-1}{2}\rfloor}{\lceil\frac{q-1}{2}\rceil}$). In this case, the Koetter-Vardy decoder gives improvements over the Guruswami-Sudan decoder. In our case, we obtain the following
	
\begin{theorem}[Generalizing results of ~\cite{CT25}, Koetter-Vardy decoder]\label{Theorem:5}
	Let $q$ be a prime integer, $n=q$ and $k\le n$. Let $S = \Iint{-z}{z} \subseteq \F_q$ and $T_{\ttau} = \{\yv = (y_i)_{i \in \Iint{1}{n}} : |\{i : y_i \in S\}| \ge \ttau n\}$ for any $\ttau \in (0,1)$. There exists a quantum algorithm that solves $\ICC(\RS_{k},T_{\ttau})$ with probability $\frac{1}{\poly(n)}$ for $\ttau = \tau - o(1)$ where $\tau$ is any real such that 
	\begin{align}\label{Eq:Thm5} (1-\frac{k}{n}) \le U(\tau,\rho), \end{align} for some quantity $U(\tau,\rho)$ specified below. Let $B = \sqrt{\frac{1 - \tau}{q - (2z+1)}}$, $A = \sqrt{\frac{\tau}{2z+1}} - B$ and $\Gamma = 2AB(2z+1) + qB^2.$ \\
	For $\rho \le 0.5$, we have 
	$$ U(\tau,\rho) = A^4   \frac{2 \rho^3 q^2}{3} + 2A^2\Gamma \rho^2q + \Gamma^2.$$
	For $\rho \ge 0.5$, we have 
	$$U(\tau,\rho) = A^4 \left(q^2 \rho^2 \left(\frac{10\rho}{3} - 4 + \frac{2}{\rho} - \frac{1}{3\rho^2}\right)\right) + \frac{2A^2\Gamma}{q}\left(2z+1 +  {l(4z+1-l)} + (2z+1-l)(q-2l-1)\right) + \Gamma^2.$$
	In particular, for $\frac{k}{n} = \frac{2}{3}$ and $\rho = 0.5$, one can show numerically that the above is satisfied even with $\tau = 1$, recovering the results~\cite{CT25}.
\end{theorem}

The proof idea will be the same but the calculations are more complicated.

\begin{proof}
	Fix $\tau$ that satisfies Equation~\ref{Eq:Thm3}. We consider the function
	$u : \F_q \rightarrow \mathbb{C}$ such that \[ \hu(\alpha) = \left\{ \begin{array}{cl} \sqrt{\frac{\tau}{|S|}} & \textrm{ if } \alpha \in S \\ \sqrt{\frac{1 - \tau}{q - |S|}} & \textrm{ if } \alpha \notin S \end{array}\right.\]
	Also, let $f = u^{\otimes n}$ and $p = |f|^2 = (|u|^2)^{\otimes n}$. The crucial calculation is captured by the following lemma
	\begin{lemma}\label{Lemma:KV}
		$\sum_{\alpha} |u|^4(\alpha) \ge U(\tau,\rho)$
	\end{lemma}
	\begin{proof}
		We defer the proof of this lemma to Appendix~\ref{Appendix:KV}.
	\end{proof}
	From this lemma, we obtain 
	$$ 1 - \frac{k}{n} \le U(\tau,\rho) \le \sum_{\alpha} |u|^4(\alpha).$$
	 We can use Proposition~\ref{Proposition:KV} with $v = |u|^2$ to say that there exists an efficient algorithm to solve $\SD(\RS_{n-k},p)$, hence $\QDP(\RS_{n-k},f)$ with probability $\frac{1}{\poly(n)}$. This gives the first condition required in Theorem~\ref{Theorem:Main}. Since states $\sum_{\alpha \in \F_q} \hu_i(\alpha) \ket{\alpha}$ are efficiently computable, the state $\sum_{\ev \in \F_q^n} f(\ev) \ket{\ev}$ is also efficiently computable. For the third condition, we know from Proposition~\ref{Proposition:2} that 
	$$ \sum_{\yv \in T_{\ttau}} |\hf(\yv)|^2 = 1 - \eta \text{ with } \eta = negl(n).$$
	We use Theorem~\ref{Theorem:Main} to conclude that there exists an efficient quantum algorithm that solves $ \ICC(\RS_{k},T_{\ttau})$ efficiently with success probability $\frac{1}{\poly(n)}$.
	\end{proof} 

\subsection{Graphical representation of our results}
As an example, we provide a few numerical results for the case $\rho = \frac{1}{2}$, which corresponds to the case studied in previous work. First we give in Figure~\ref{Figure:1} a plot presenting existing results as well as our new theorems. 
\begin{figure}[!ht]
	\centering
	\includegraphics[width = 0.8\textwidth]{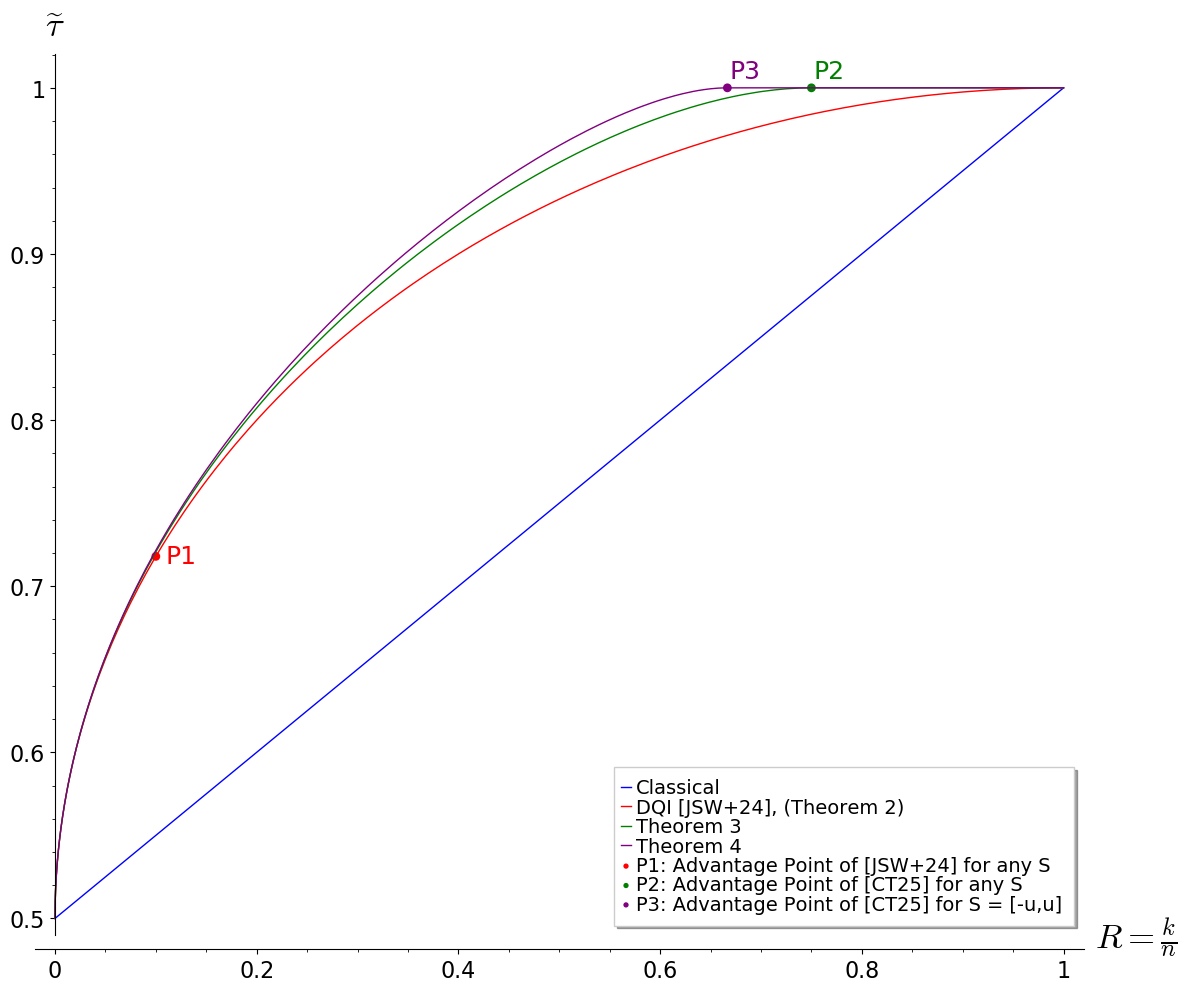}
	\caption{Ratio of satisfied constraints $\ttau$ as a function of the code rate $R = \frac{k}{n}$, for  $\rho = \frac{|S|}{q} = \frac{1}{2}$.}
	\label{Figure:1}
\end{figure}

The authors of~\cite{JSW+24} considered for OPI an instance where the best efficient classical algorithm solve 55\% of the constraints. We also show, for different decoders, what are the parameters that achieve the highest advantage for this classical threshold.

	\begin{table}[h]
	\centering
	\begin{tabular}{|c|c||c|c|c|c||l|}
		\hline
		$R = \frac{k}{n}$ & $\rho = \frac{|S|}{q}$ & $\ttau$(Classical) & $\ttau$(DQI) & $\ttau$(Th. 3)& $\ttau$(Th. 4)& Comments \\
		\hhline{|=|=||=|=|=|=||=|}
		0.1 & 0.5 & 0.55 & 0.718 & 0.721 & 0.722 & P1 point \\
		0.75 & 0.5 & 0.875 & 0.984 & 1 & 1 & P2 point \\
		$\frac{2}{3}$ & 0.5 & 0.833 & 0.971 & 0.994 & 1 & P3 point \\
		\hline 
		0.234 & 0.413 & 0.55 & 0.749 & 0.760 & 0.763 & Opt. $\ttau$(DQI) for $\ttau$(Classical) = 0.55 \\
		0.259 & 0.393 & 0.55 & 0.748 & 0.761 & 0.765 & Opt. $\ttau$(Th. 3) for $\ttau$(Classical) = 0.55 \\
		0.267 & 0.386 & 0.55 & 0.748 & 0.761 & 0.765 & Opt. $\ttau$(Th. 4) for $\ttau$(Classical) = 0.55 \\
		\hline
	\end{tabular}
	\caption{Different achievable thresholds $\ttau$ given parameters $R,\rho$}
	\label{tab:your_label}
\end{table}


\COMMENT{\newpage
\section{Conclusion}
In this work, we provided a simple and general framework that generalizes both the results of~\cite{JSW+24} where we recover their results and extend them to the case of decoders with errors, and the results of~\cite{CT25}, that we extend to the case where we want to satisfy only a fraction $\ttau$ of the constraints. 

In our algorithm analysis, our main insight was not to try to compare a real setting with an ideal setting without errors but rather just focus on the final problem we want to solve and on generic properties of the function $f$ we use. This analysis also ensures that if we have a decoder for $\SD(\Hm,|f|^2)$, then a solution to the corresponding $\ICC(\C^\bot,T)$ will necessarily exist without needing any extra property or conjecture on the code. 

In the general $q$-ary setting, the authors of~\cite{JSW+24} use fairly elaborate techniques to exhibit the error functions for their algorithm. In this work, we construct a natural function that essentially matches their optimal bound. This function is actually a product function which makes is very simple to analyze. 
}

\textbf{Acknowledgments.} The author thanks Noah Shutty for valuable feedback. We acknowledge funding from the French PEPR integrated projects EPIQ (ANR-22-PETQ-007), PQ-TLS (ANR-22-PETQ-008) and HQI (ANR-22-PNCQ-0002) all part of plan France 2030.

\newpage
\bibliography{paper}
\bibliographystyle{alpha}
\newpage
\begin{appendix}

\section{Rewriting~\cite{CH25} for codes}\label{Appendix:MainProof}
In this work, we prove Theorem~\ref{Theorem:Main}. As we mentioned, this is a direct adaptation of the reduction theorem in~\cite{CH25} that we reproduce verbatim, only by changing all lattice concepts with the corresponding code concepts. So we start from a $q$-ary code $\C$ with associated generating matrix $\Gm$ and we construct an algorithm to solve $\ICC(\C^\bot,T)$ for some $T \subseteq \F_q^n$. 

\subsection{Characterization of quantum algorithms for $\QDP$}

A quantum algorithm for $\QDP(\C,f)$ can be described by a unitary $U$ (that depends on $\C$ and $f$) such that 
$$ \forall \sv \in \F_q^k, \ U \ket{\psi_\sv}\ket{0} = \sum_{\sv' \in \F_q^k} \gamma_{\sv,\sv'} \ket{\sv'}\ket{\wpsi_{\sv,\sv'}}, \text{ for some unit vectors } \ket{\wpsi_{\sv,\sv'}} \text{ and } \gamma_{\sv,\sv'} \in \mathbb{C},$$
and the result is obtained by measuring the first register. 
The success probability of this algorithm for each $\sv$ is $p_\sv =|\gamma_{\sv,\sv}|^2$ and the overall success probability is $p = \frac{1}{q^k} \sum_{\sv} |\gamma_{\sv,\sv}|^2$. We first prove that any such quantum algorithm can be symmetrized in the sense that each $\gamma_{\sv,\sv}$ is equal to $\sqrt{p}$.
\begin{proposition}\label{Proposition:Uniform}
	Let $\aa$ be an efficient quantum algorithm for $\QDP(\C,f)$ that succeeds with probability $p$. There exists an efficiently computable unitary $U$ such that 
	$$ \forall \sv \in \F_q^k, \ U \ket{\psi_\sv}\ket{0} = \sum_{\sv' \in \F_q^k} \gamma'_{\sv,\sv'} \ket{\sv'}\ket{\wpsi''_{\sv,\sv'}}, \text{ for some unit vectors } \ket{\wpsi''_{\sv,\sv'}} \text{and each } \gamma'_{\sv,\sv} = \sqrt{p}.$$
\end{proposition}
\begin{proof}
	The idea is to use the symmetries inherent to the states $\ket{\psi_\sv}$. We present a first algorithm that succeeds with probability $p$ for each input state $\ket{\psi_\sv}$. Consider the shift unitaries $S_\zv : \ket{\xv} \rightarrow \ket{\xv + \zv}$ for $\xv,\zv \in \F_q^n$ which are efficiently computable. Notice that $\forall \sv,\tv \in \F_q^k$, we have $\ket{\psi_\tv} = S_{(\tv - \sv)\Gm } \ket{\psi_\sv}$. We consider the following algorithm 
	\begin{enumerate}
		\item Given input $\ket{\psi_\sv}$, construct  
		$$ \ket{\Omega_1} = \frac{1}{\sqrt{q^k}}\sum_{\tv \in \F_q^k}S_{\tv\Gm }\ket{\psi_\sv}\ket{0}\ket{\tv} = \sum_{\tv \in \F_q^k}\ket{\psi_{\sv + \tv}}\ket{0}\ket{\tv}.$$
		\item  Apply $U$ on the first two register to obtain 
		$$ \ket{\Omega_2} = \frac{1}{\sqrt{q^k}}\sum_{\tv \in \F_q^k} \sum_{\sv' \in \F_q^k} \gamma_{\sv+\tv,\sv'} \ket{\sv'} \ket{\wpsi_{\sv + \tv,\sv'}}\ket{\tv}.$$
		\item We subtract the value from the third register in the first register to obtain
		\begin{align*}\ket{\Omega_3} & = \frac{1}{\sqrt{q^k}}\sum_{\tv \in \F_q^k} \sum_{\sv' \in \F_q^k} \gamma_{\sv+\tv,\sv'} \ket{\sv' - \tv} \ket{\wpsi_{\sv + \tv,\sv'}}\ket{\tv} \\
			& = \frac{1}{\sqrt{q^k}}\sum_{\sv' \in \F_q^k}\sum_{\tv \in \F_q^k} \gamma_{\sv + \tv,\sv'+\tv}\ket{\sv'} \ket{\wpsi_{\sv + \tv,\sv' + \tv}}\ket{\tv} \\
		\end{align*}
	\end{enumerate}
	If we measure the first register, we obtain $\sv$ with probability $\frac{1}{q^k} \sum_{\tv} |\gamma_{\sv + \tv,\sv + \tv}|^2 = p$ which is independent of $\sv$. If we perform the above algorithm fully coherently, we obtain a quantum unitary $U'$ such that 
	$$ \forall \sv \in \F_q^k, \ U' \ket{\psi_\sv}\ket{0} = \sum_{\sv' \in \F_q^k} \gamma'_{\sv,\sv'} \ket{\sv'}\ket{\wpsi'_{\sv,\sv'}}, \text{ for some unit vectors } \ket{\wpsi_{\sv,\sv'}} \text{ and each } |\gamma'_{\sv,\sv}| = \sqrt{p}.$$
	In order to conclude, we just have to put the potential phases of $\gamma'_{\sv,\sv}$ into the second register so if we define $\ket{\wpsi''_{\sv,\sv'}} = \frac{\gamma_{\sv,\sv'}}{|\gamma_{\sv,\sv'}|}\ket{\wpsi'_{\sv,\sv'}}$, we can indeed write 
	$$  U' \ket{\psi_\sv}\ket{0} = \sum_{\sv' \in \F_q^k} \gamma'_{\sv,\sv'} \ket{\sv'}\ket{\wpsi''_{\sv,\sv'}} \quad \text{with each } \gamma'_{\sv,\sv'} = \sqrt{p}. \qedhere$$
\end{proof}

\subsection{Description of the algorithm}

We present now a detailed description of our algorithm. We assume we have access to a quantum algorithm for $\QDP$ that satisfies the condition of Proposition~\ref{Proposition:Uniform} and we use it to construct an algorithm for $\ICC$.

\begin{tcolorbox}[colback=gray!10!white, colframe=gray!75!black, title=Algorithm 1: Quantum algorithm based on Regev's reduction for \ICC, fonttitle=\bfseries\large, boxrule=1.5pt, arc=4mm, width=\linewidth, before skip=10pt, after skip=10pt]
	
	\textbf{Input:} We start from a $q$-ary code $\C$ of length $k$ and dimension $n$ with a associated generating matrix $\Gm \in \F_q^{k \times n}$. Let $T \subseteq \F_q^n$ and $f : \F_q^n \rightarrow \mathbb{C}$ such that $\norm{f}_2 = 1$. For each $\sv \in \F_q^k$, we write $\ket{\psi_{\sv}} = \sum_{\ev \in \F_q^n} f(\ev)\ket{\sv\Gm + \ev}$.  Assume we have a codeword independent quantum algorithm $\aa$ that solves $\QDP(\C,f)$ with probability $P_{Dec}$. This means we have a quantum unitary 
	$$ U \ket{\psi_\sv}\ket{0} = \sum_{\sv' \in \F_q^k} \gamma_{\sv,\sv'} \ket{\wpsi_{\sv,\sv'}}\ket{\sv'}, \ \ \text{where } \forall \sv \in \F_q^k, \ \gamma_{\sv,\sv} = \sqrt{P_{Dec}}
	\text{ and each } \ket{\wpsi_{\sv,\sv'}} \text{ is a unit vector}. $$
	Finally, we are given a random $\uv \in \F_q^k$. \\
	\textbf{Goal:} Find $\yv \in \C^{\bot}_{\uv} \cap T$, where $\C^\bot_{\uv} = \{\yv \in \F_q^n :  {\Gm}\trsp{\yv} = \trsp{\uv}\}$ \\ \\
	\textbf{Execution of the algorithm:}
	\begin{enumerate}
		\item First construct the state
		$\frac{1}{\sqrt{q^k}} \sum_{\sv \in \F_q^k} \chi_{-\uv}(\sv) \ket{\psi_\sv}\ket{0}\ket{\sv}$.
		\item Perform the operation
		\begin{align*}
			\frac{1}{\sqrt{q^k}} \sum_{\sv \in \F_q^k}\chi_{-\uv}(\sv) \ket{\psi_{\sv}}\ket{0}\ket{\sv} & \overset{\textcircled{\scalebox{0.7}{A}}}{\mathlarger{\mathlarger{\rightarrow}}} \frac{1}{\sqrt{q^k}}\sum_{\sv,\sv' \in \F_q^k} \chi_{-\uv}(\sv) \gamma_{\sv,\sv'} \ket{\wpsi_{\sv,\sv'}}\ket{\sv'}\ket{\sv - \sv'}
		\end{align*}
		Here, $\textcircled{\scalebox{0.7}{A}}$ is done by applying $U$ on the first two registers and then substracting the second register from the third register. 
		\item Measure the third register. If we don't obtain $\mathbf{0^n}$, start again from step $1$. Otherwise, we obtain the state $\frac{1}{\sqrt{q^k}}\sum_{\sv \in \F_q^k} \chi_{-\uv}(\sv) \ket{\wpsi_{\sv,\sv}}\ket{\sv}\ket{0}$.
		\item Discard the third register and apply $U^\dagger$ on the first two registers. The resulting state is
		$$\ket{\Phi_{\uv}} = \frac{1}{\sqrt{q^k}}\sum_{\sv \in \F_q^k} \chi_{-\uv}(\sv)\sqrt{P_{Dec}}\ket{\psi_\sv} + \chi_{-\uv}(\sv)\sqrt{1 - P_{Dec}} \ket{Z_\sv} \ \ \text{ for some unit vector } \ket{Z_\sv} \bot \ket{\psi_\sv}.$$ 
		\item Compute $\ket{\widehat{\Phi_{\uv}}}$ and measure in the computational basis. Output the outcome of the measurement.
	\end{enumerate}
\end{tcolorbox}
\subsection{First analysis and running time of the algorithm}
We first provide some details over each step of the algorithm. 
\begin{enumerate}
	\item The initialization step of the algorithm can be done as follows
	$$ \sum_{\sv \in \F_q^k} \ket{\sv} \otimes \sum_{\ev \in \F_q^k} f(\ev) \ket{\ev} \overset{\textcircled{\scalebox{0.7}{1}}}{\mathlarger{\mathlarger{\rightarrow}}}  \sum_{\substack{\sv \in \F_q^k \\ \ev \in \F_q^n}} f(\ev)\ket{\sv}\ket{\sv \Gm + \ev} = \sum_{\sv \in \F_q^k} \ket{\sv}\ket{\psi_{\sv}},$$
	which corresponds to the initial state by adding a $\ket{\zerov}$ register and reordering. In $\textcircled{\scalebox{0.7}{1}}$, we use the fact that $\sv \rightarrow \sv\Gm$ is easily computable and apply this operation coherently. We then need to compute $\sum_{\ev \in \F_q^n} f(\ev) \ket{\ev}$ which take some time $T_{Sampl}$. In practice, $f$ is chosen such that this state can be computed efficiently. The running time of this is therefore in $O(T_{Sampl} + \poly(n,\log(q)))$.
	\item In step $3$, before the measurement, we have the state 
	\begin{align*}\frac{1}{\sqrt{q^k}}\sum_{\sv,\sv' \in \F_q^k} \chi_{-\uv}(\sv) \gamma_{\sv,\sv'}& \ket{\wpsi_{\sv,\sv'}} \ket{\sv'}\ket{\sv - \sv'} = \\ 
		& \frac{1}{\sqrt{q^k}}\left(\sum_{\sv} \chi_{-\uv}(\sv) \sqrt{P_{Dec}}  \ket{\wpsi_{\sv,\sv}}\ket{\sv}\ket{0} + \sum_{\sv,\sv' \neq \sv} \chi_{-\uv}(\sv) \sqrt{1-P_{Dec}}  \ket{\wpsi_{\sv,\sv'}}\ket{\sv'}\ket{\sv - \sv'}\right).
	\end{align*}
	which means that we successfully measure $\zerov$ in the last register with probability $P_{Dec}$ and that conditioned on this outcome, the resulting state is indeed $\frac{1}{\sqrt{q^k}}\sum_{\sv} \chi_{-\uv}(\sv)   \ket{\wpsi_{\sv,\sv}}\ket{\sv}\ket{0}$. This means we have to repeat steps $1$ to $3$ $O(\frac{1}{P_{Dec}})$ times. Moreover, step $2$ requires to compute $U$, which takes time $T_{Dec}$ which is the running time of the decoder. From there, we conclude that the time required for this algorithm to successfully pass step $3$ is 
	$$O\left(\frac{1}{P_{Dec}} \left(\Time_{Dec} + \Time_{Sampl}\right) + \poly(n,\log(q))\right).$$ 
	\item In order to see step $4$, we start from $U \ket{\psi_\sv}\ket{0} = \sum_{\sv' \in \F_q^k} \gamma_{\sv,\sv'} \ket{\wpsi_{\sv,\sv'}}\ket{\sv'}$ which implies 
	$$ \bra{\wpsi_{\sv,\sv}}\bra{\sv} U \left(\ket{\psi_{\sv}}\ket{\zerov}\right) = \bra{\psi_{\sv}}\bra{\zerov} U^\dagger  \left(\ket{\wpsi_{\sv,\sv}}\ket{\sv}\right) = \gamma_{\sv,\sv} = \sqrt{P_{Dec}}.$$
	This means that for each $\sv \in \F_q^k$,  we can indeed write 
	$$ U^\dagger(\ket{\wpsi_{\sv,\sv}}\ket{\sv}) = \sqrt{P_{Dec}} \ket{\psi_{\sv}} + \sqrt{1 - P_{Dec}} \ket{Z_\sv},$$
	for some unit vector $\ket{Z_{\sv}}$ orthogonal to $\ket{\psi_{\sv}}$, which justifies step $4$ of the algorithm. Finally, in step $5$, we have to perform $n$ quantum Fourier transforms in $F_q$ and measure, which takes time $\poly(n,\log(q))$.
\end{enumerate}
From this analysis, we can conclude that the running time of the algorithm satisfies 
$$\Time = O\left(\frac{1}{P_{Dec}} \left(\Time_{Dec} + \Time_{Sampl}\right) + \poly(n,\log(q))\right).$$
The trickier part will be to argue about the success probability of the algorithm, which is the goal of the following section.

\subsection{Proof of main theorem}
We use the notations as in the beginning of the section. We prove the following
\begin{theorem}[Theorem~\ref{Theorem:Main} restated]
	Let $\C$ be a $q$-ary linear code of dimension $k$ and length $n$ and let $\Gm \in \F_q^{k \times n}$ be a generating matrix of $\C$. Let $T \subseteq \F_q^n$. Let $f : \F_q^n \rightarrow \mathbb{C}$ with $\norm{f}_2 = 1$. Assume that 
	\begin{enumerate}\setlength\itemsep{-0.3em}
		\item We have a (classical or quantum) codeword independent algorithm $\aa_{Dec}$ that solves $\SD(\C,|f|^2)$ in time $\Time_{Dec}$ and succeeds with probability $P_{Dec}$.
		\item The state $\sum_{\ev \in \F_q^n} f(\ev) \ket{\ev}$ is constructible in time $T_{Sampl}$.
		\item $\sum_{\yv \in T}|\hf(\yv)|^2 = 1 - \eta$.
	\end{enumerate}
	Then there exists a quantum algorithm that solves $\ICC(\C^\bot,T)$ with probability
	$$ P \ge P_{Dec}(1-\eta) - 2\sqrt{\eta P_{Dec}(1-P_{Dec})},$$
	and runs in time 
	$$\Time = O\left(\frac{1}{P_{Dec}} \left(\Time_{Dec} + \Time_{Sampl}\right) + \poly(n,\log(q))\right).$$
\end{theorem}
The running time of the algorithm has been discussed in the previous section so we just need to prove the success probability. 
We fix $\uv \in \F_q^{k}$, and let $p_{\uv}$ be the probability that the algorithm outputs an element $\yv \in \C^\bot_{\uv} \cap T$ given this $\uv$. We write 
$\ket{\widehat{Z_\sv}} = \sum_{\yv \in \F_q^n} z_{\sv,\yv} \ket{\yv}$ and have 
\begin{lemma}
	$p_{\uv} = \sum_{\yv \in \C^\bot_{\uv} \cap T} \left|\sqrt{\Pdec} \sqrt{q^k} \hf(\yv) + \frac{\sqrt{1 - \Pdec}}{\sqrt{q^k}} \sum_{\sv \in \F_q^k} \omega^{- \sv \cdot \uv} \hzsy \right|^2$.
\end{lemma}
\begin{proof}
	In order to compute $p_{\uv}$, we have to compute $\ket{\widehat{\Phi_{\uv}}}$. We have 
	$$\ket{\widehat{\Phi_{\uv}}} = \frac{1}{\sqrt{q^k}}\sum_{\sv \in \F_q^k} \chi_{\uv}(\sv)\sqrt{P_{Dec}} \ket{\widehat{\psi_\sv}} + \chi_{\uv}(\sv)\sqrt{1 - P_{Dec}} \ket{\widehat{Z_\sv}}.$$
	We first write 
	\begin{align*}
		\ket{\widehat{\psi_{\sv}}} & = \frac{1}{\sqrt{q^k}}\sum_{\yv \in \F_q^n}\sum_{\ev \in \F_q^n} \chi_{\yv}(\sv \Gm + \ev) f(\ev)\ket{\yv} \\
		& = \frac{1}{\sqrt{q^k}}\sum_{\yv \in \F_q^n}  \chi_{\yv}(\sv \Gm)  \sum_{\ev \in \F_q^n} \chi_{\yv}(\ev) f(\ev)\ket{\yv} \\
		& = \sum_{\yv \in \F_q^n}  \chi_{\yv \trsp{\Gm}}(\sv)  \hf(\yv) \ket{\yv} \\
		& = \sum_{\uv' \in \F_q^k} \chi_{\uv'}(\sv)  \sum_{\yv \in \C^{\bot}_{\uv'}}  \hf(\yv) \ket{\yv}
	\end{align*}
	which gives
	$$ \sum_{\sv \in \F_q^k} \chi_{-\uv}(\sv) \ket{\widehat{\psi_{\sv}}} = \sum_{\sv \in \F_q^k} \sum_{\uv' \in \F_q^k} \chi_{(\uv' - \uv)}(\sv) \sum_{\yv \in \C^\bot_{\uv'}} \hf(\yv) \ket{\yv} = q^k \sum_{\yv \in \C^\bot_{\uv}} \hf(\yv) \ket{\yv}.$$
	From there, we have 
	$$ \ket{\widehat{\Phi_{\uv}}} = \sqrt{q^k{P_{Dec}}} \sum_{\yv \in \C^\bot_{\uv}} \hf(\yv)\ket{\yv} + \sqrt{\frac{1-P_{Dec}}{q^k}}\sum_{\yv \in \F_q^n} \sum_{\sv \in \F_q^k} \chi_{-\uv}(\sv) z_{\sv,\yv}\ket{\yv}. $$	
	The algorithm computes this state and measures in the computational basis. The probability to output an element of $\C^\bot_{\uv} \cap T$ is therefore 
	$$ p_{\uv} = \sum_{\yv \in \C^\bot_{\uv} \cap T} \left|\sqrt{q^k{P_{Dec}}}\hf(\yv) + \sqrt{\frac{1-P_{Dec}}{q^k}}\sum_{\sv \in \F_q^k} \chi_{-\uv}(\sv) z_{\sv,\yv}\right|^2. \qedhere $$
\end{proof}
We can now proceed to the main proof of this section
\begin{proposition}
	$P = \E_{\uv \Unif \F_q^k} \left[p_{\uv}\right] \ge P_{Dec}(1-\eta) - 2\sqrt{\eta P_{Dec}(1-P_{Dec})}$. 
\end{proposition}
\begin{proof}
	We write 
	\begin{align*}
		p_{\uv} & = \sum_{\yv \in \C^\bot_{\uv} \cap T} \left|\sqrt{\Pdec} \sqrt{q^k} \hf(\yv) + \frac{\sqrt{1 - \Pdec}}{\sqrt{q^k}} \sum_{\sv \in \F_q^k} \chi_{-\uv}(\sv) \hzsy \right|^2 \\
		& = \sum_{\yv \in \C^\bot_{\uv} \cap T} \left({\Pdec}{q^k} |\hf(\yv)|^2 + \frac{{1 - \Pdec}}{{q^k}} \left|\sum_{\sv \in \F_q^k} \chi_{-\uv}(\sv) \hzsy \right|^2  + 2Re\left(\sqrt{\Pdec} \sqrt{q^k} \hf(\yv)\frac{\sqrt{1 - \Pdec}}{\sqrt{q^k}} \sum_{\sv \in \F_q^k} \chi_{\uv}(\sv) \ohzsy\right)\right)
	\end{align*}
	where we used $|a+b|^2 = (a+b)(\overline{a} + \overline{b}) = |a|^2 + |b|^2 + 2Re(a\overline{b})$. We now bound each term separately. We first write 
	\begin{align*}
		\E_{\uv \Unif \F_q^k} \left[\sum_{\yv \in \C^\bot_{\uv} \cap T} \left({\Pdec}{q^k} |\hf(\yv)|^2\right)\right] & = \Pdec(1 - \eta) \\
		\forall \uv \in \F_q^k, \ \sum_{\yv \in \C^\bot_{\uv} \cap T} \left(\frac{{1 - \Pdec}}{{q^k}} \left|\sum_{\sv \in \F_q^k} \chi_{-\uv}(\sv)\hzsy \right|^2\right) & \ge 0 \\
		\sum_{\yv \in \C^\bot_{\uv} \cap T} \left(2Re\left(\sqrt{\Pdec} \sqrt{q^k} \hf(\yv)\frac{\sqrt{1 - \Pdec}}{\sqrt{q^k}} \sum_{\sv \in \F_q^k} \chi_{\uv}(\sv)\ohzsy\right)\right) & = 2\sqrt{\Pdec(1-\Pdec)} Re\left(\sum_{\substack{\yv \in \C^\bot_{\uv} \cap T \\ \sv \in \F_q^k}} \hf(\yv)\chi_{\uv}(\sv)\ohzsy\right)
	\end{align*}
	
	From there, we write 
	\begin{align}\label{Eq:FinalProof} p = \E_{\uv \Unif \F_q^k} \left[p_{\uv}\right] \ge \Pdec(1 - \eta) + 2\sqrt{\Pdec(1-\Pdec)} \E_{\uv \Unif \F_q^k} \left[Re\left(\sum_{\substack{\yv \in \C^\bot_{\uv} \cap T \\ \sv \in \F_q^k}} \hf(\yv)\chi_{\uv}(\sv) \ohzsy\right)\right]. \end{align}
	In order to conclude, we prove the following lemma
	\begin{lemma}
		$$\forall \sv \in \F_q^k, \  \left|\sum_{\uv \in \F_q^k}\sum_{\yv \in \C^\bot_{\uv}} \hf(\yv)\chi_{\uv}(\sv)\ohzsy\right| \le \sqrt{\eta}.$$
	\end{lemma}
	\begin{proof}
		We start from the equality $\braket{Z_\sv}{\psi_\sv} = \braket{\widehat{Z_\sv}}{\widehat{\psi_\sv}} = 0$ for each $\sv \in \F_q^k$, which can be rewritten
		\begin{align*}
			\sum_{\uv \in \F_q^k} \sum_{\yv \in \C^\bot_{\uv}} \chi_{\uv}(\sv)\hf(\yv)\ohzsy = 0.
		\end{align*}
		This implies 
		\begin{align*}
			\left|\sum_{\uv \in \F_q^k} \sum_{\yv \in \C^\bot_{\uv} \cap T} \chi_{\uv}(\sv)\hf(\yv)\ohzsy\right| & = \left|\sum_{\uv \in \F_q^k} \sum_{\yv \in \C^\bot_{\uv} \cap \overline{T}} \chi_{\uv}(\sv)\hf(\yv)\ohzsy\right| \\
			& = \left|\sum_{\yv \notin T} \chi_{\yv \trsp{\Gm}}(\sv) \hf(\yv)\ohzsy\right|\\
			& \le \sqrt{\sum_{\yv \notin T} |\chi_{\yv \trsp{\Gm}}(\sv) \hf(\yv)|^2}\sqrt{\sum_{\yv \notin T} |\hzsy|^2} \\
			& \le \sqrt{\eta}\sqrt{1} = \sqrt{\eta},
		\end{align*}
		where we used the fact that $\ket{Z_\sv}$ is a unit vector.
	\end{proof}
	We can now conclude our main proof. We have 
	\begin{align*}
		\E_{\uv \Unif \F_q^k} \left[Re\left(\sum_{\substack{\yv \in \C^\bot_{\uv} \cap T \\ \sv \in \F_q^k}} \hf(\yv)\chi_{\uv}(\sv)\ohzsy\right)\right] & = \frac{1}{q^k} \sum_{\sv \in \F_q^k}  Re\left(\sum_{\uv \in \F_q^k} \sum_{\yv \in \C^\bot_{\uv}} \hf(\yv)\chi_{\uv}(\sv) \ohzsy\right) \\
		& \ge  -\frac{1}{q^k} \sum_{\sv \in \F_q^k}  \left|\sum_{\uv \in \F_q^k} \sum_{\yv \in \C^\bot_{\uv}} \hf(\yv)\chi_{\uv}(\sv) \ohzsy\right| \\
		& \ge -\frac{1}{q^k} \sum_{\sv \in \F_q^k} \sqrt{\eta} \\
		& = -\sqrt{\eta}
	\end{align*} 
	Plugging this inequality into Equation~\ref{Eq:FinalProof}, we get 
	$$ P \ge \Pdec(1-\eta) - 2\sqrt{\Pdec(1-\Pdec)}\sqrt{\eta}.$$ 
\end{proof}

	\section{Proof of Lemma~\ref{Lemma:KV}}\label{Appendix:KV}
	\begin{proposition}
		Let $u : \F_q \rightarrow \mathbb{C}$ be the function such that \[ \hu(\alpha) = \left\{
		\begin{array}{cl}
			\sqrt{\frac{\tau}{|S|}} & \textrm{ if } \alpha \in S \\
			\sqrt{\frac{1 - \tau}{q - |S|}} & \textrm{ if } \alpha \notin S
		\end{array}\right.\]
		We have 
		\[ \sum_{\alpha} (u \star u)^2(\alpha) = \sum_{\alpha} |u(\alpha)|^4 = U(\tau,\rho)\]
	\end{proposition} 
	\begin{proof}
		We write $\hu = A \one_{\Iint{-z}{z}} + B \one_{\F_q}$ with $B = \sqrt{\frac{1 - \tau}{q - |S|}}$ and $A = \sqrt{\frac{\tau}{|S|}} - B$. We then write 
		\[ \hu \star \hu = A^2 (\one_{\Iint{-z}{z}} \star \one_{\Iint{-z}{z}}) + AB (\one_{\Iint{-z}{z}} \star \one_{\F_q}) + AB (\one_{F_q} \star \one_{\Iint{-z}{z}}) + B^2 (\one_{\F_q} \star \one_{\F_q})\]
		Now, notice that 
		\[ (\one_{\Iint{-z}{z}} \star \one_{\F_q}) = (\one_{F_q} \star \one_{\Iint{-z}{z}}) = (2z+1) \one_{\F_q},\]
		and $(\one_{\F_q} \star \one_{\F_q}) = q \one_{\F_q}$. From there, we write 
		\[ (\hu \star \hu)(\alpha) = A^2 (\one_{\Iint{-z}{z}} \star \one_{\Iint{-z}{z}})(\alpha) + 2AB(2z+1) + qB^2.\]
		We define $\Gamma = 2AB(2z+1) + qB^2$. We can now write 
		\[
		(\hu \star \hu)^2(\alpha) = A^4  (\one_{\Iint{-z}{z}} \star \one_{\Iint{-z}{z}})^2(\alpha) + 2A^2\Gamma (\one_{\Iint{-z}{z}} \star \one_{\Iint{-z}{z}})(\alpha) + \Gamma^2.\] 
		which gives us 
		\[
		\sum_{\alpha \in \F_q}  (\hu \star \hu)^2(\alpha) = A^4  \left(\sum_{\alpha \in \F_q} (\one_{\Iint{-z}{z}} \star \one_{\Iint{-z}{z}})^2(\alpha)\right) + 2A^2\Gamma \left(\sum_{\alpha \in \F_q} (\one_{\Iint{-z}{z}} \star \one_{\Iint{-z}{z}})(\alpha)\right) + q\Gamma^2.\]
		In order to conclude we use the analysis of~\cite{CT25}. Let $\rho = \frac{|S|}{q} = \frac{2z+1}{q}$. We distinguish two cases
		\paragraph{Case 1: $\rho \le \frac{1}{2}$.} In this case, we have 
		\begin{align*}
			\left(\one_{\Iint{-z}{z}} \star \one_{\Iint{-z}{z}}\right)(\alpha) = \left\{
			\begin{array}{cl}
				2z + 1 - |\alpha| & \textrm{ if } \alpha \in \Iint{-2z}{2z} \\
				0 & \textrm{ otherwise }
			\end{array}\right.
		\end{align*}
		which gives
		\begin{align*}
			\sum_{\alpha \in \F_q} \left(\one_{\Iint{-z}{z}} \star \one_{\Iint{-z}{z}}\right)(\alpha) = 2z+1 + 2\sum_{\alpha = 1}^{2z} (2z+1 - \alpha) = (2z+1) + {2z(2z+1)} = (2z+1)^2 = \rho^2q^2.
		\end{align*}
		Moreover, we know from~\cite{CT25} that $\sum_{\alpha \in \F_q} \left(\one_{\Iint{-z}{z}} \star \one_{\Iint{-z}{z}}\right)^2(\alpha) \ge \frac{2 \rho^3 q^3}{3}$. From there, we can conclude 
		\[
		\sum_{\alpha \in \F_q} (\hu \star \hu)^2(\alpha) \ge A^4   \frac{2 \rho^3 q^3}{3} + 2A^2\Gamma \rho^2q^2 + q\Gamma^2.\]
		and 
		\begin{align}
			\sum_{\alpha \in \F_q} |u(\alpha)|^4 = \frac{1}{q} \sum_{\alpha \in \F_q} (\hu \star \hu)^2(\alpha) \ge A^4   \frac{2 \rho^3 q^2}{3} + 2A^2\Gamma \rho^2q + \Gamma^2.
		\end{align}	
		\paragraph{Case 2: $\rho \ge \frac{1}{2}$.} Let $l = q - (2z+1).$ We have 
		\begin{align*}
			\left(\one_{\Iint{-z}{z}} \star \one_{\Iint{-z}{z}}\right)(\alpha) = \left\{
			\begin{array}{cl}
				2z + 1 - |\alpha| & \textrm{ if } \alpha \in \Iint{-l}{l} \\
				2z + 1 - l & \textrm{ otherwise }
			\end{array}\right.
		\end{align*}
		which gives 
		\begin{align*}
			\sum_{\alpha \in \F_q} \left(\one_{\Iint{-z}{z}} \star \one_{\Iint{-z}{z}}\right)(\alpha) & = 2z+1 + 2\left(\sum_{\alpha = 1}^{l} (2z+1 - \alpha)\right) + (2z+1-l)(q-2l-1) \\
			& = 2z+1 +  {l(4z+1-l)} + (2z+1-l)(q-2l-1)
		\end{align*} 
		Moreover, we know from~\cite{CT25} that 
		\[\sum_{\alpha \in \F_q} \left(\one_{\Iint{-z}{z}} \star \one_{\Iint{-z}{z}}\right)^2(\alpha) \ge q^3 \rho^2 \left(\frac{10\rho}{3} - 4 + \frac{2}{\rho} - \frac{1}{3\rho^2}\right).\]
		which gives 
		\begin{align*}
		\sum_{\alpha \in \F_q} (\hu \star \hu)^2(\alpha) \ge A^4 \left(q^3 \rho^2 \left(\frac{10\rho}{3} - 4 + \frac{2}{\rho} - \frac{1}{3\rho^2}\right)\right) + 2A^2\Gamma\left(2z+1 +  {l(4z+1-l)} + (2z+1-l)(q-2l-1)\right) + q\Gamma^2
		\end{align*}
		and hence 
		$$ 
		\sum_{\alpha \in \F_q} |u(\alpha)|^4 =  A^4 \left(q^2 \rho^2 \left(\frac{10\rho}{3} - 4 + \frac{2}{\rho} - \frac{1}{3\rho^2}\right)\right) + \frac{2A^2\Gamma}{q}\left(2z+1 +  {l(4z+1-l)} + (2z+1-l)(q-2l-1)\right) + \Gamma^2.$$
	\end{proof}

\COMMENT{
	\begin{figure}[h]
		\includegraphics[width = 0.95\textwidth]{aplot2.png}
	\end{figure}
	\begin{figure}[h]
		\includegraphics[width = 0.95\textwidth]{aplot3.png}
	\end{figure}
}
\end{appendix}
\end{document}